\documentclass[10pt]{article}
\usepackage{graphicx} % Required for inserting images
\usepackage{caption}
\usepackage{enumitem}
\usepackage[a4paper, total={6in, 8in}]{geometry}
\usepackage{amsfonts}
\usepackage{tikz} 
\usepackage{hyperref}
\usepackage{algorithm}
\usepackage{amsthm}
\usepackage{booktabs}
\usepackage{amsmath,amssymb,amsthm,fullpage}
\usepackage{xparse}
\usetikzlibrary{arrows.meta,calc,positioning,decorations.pathreplacing}
\usepackage{cleveref}
\usepackage[super]{nth}
\usepackage[noend]{algpseudocode}
\usepackage[english]{babel}
\usepackage{thmtools}
\usepackage{thm-restate}
\newtheorem{theorem}{Theorem}
\newtheorem{observation}[theorem]{Observation}

\newtheorem{lemma}[theorem]{Lemma}
\newtheorem{corollary}[theorem]{Corollary}

\theoremstyle{definition}
\newtheorem{definition}{Definition}
\newtheorem*{remark}{Remark}

\newcommand{\reals}{\mathbb{R}}

\DeclareMathOperator{\range}{range}

\usepackage[parfill]{parskip}

\usetikzlibrary{arrows.meta,calc,positioning,decorations.pathreplacing}

\NewDocumentCommand{\interval}{O{0} m m O{} O{black} O{} O{}}{%
  \draw[very thick,#5] (#2,#1) -- (#3,#1);
  \draw[very thick,#5] (#2,#1+0.12) -- (#2,#1-0.12);
  \draw[very thick,#5] (#3,#1+0.12) -- (#3,#1-0.12);
  \IfNoValueF{#4}{\node[above=2pt] at ($(#2,#1)!0.5!(#3,#1)$) {#4};}
  \IfNoValueF{#6}{\node[below] at (#2,#1-0.12) {#6};}
  \IfNoValueF{#7}{\node[below] at (#3,#1-0.12) {#7};}
}

\NewDocumentCommand{\intervalthin}{O{0} m m O{} O{black} O{} O{}}{%
  \draw[#5] (#2,#1) -- (#3,#1);
  \draw[#5] (#2,#1+0.12) -- (#2,#1-0.12);
  \draw[#5] (#3,#1+0.12) -- (#3,#1-0.12);
  \IfNoValueF{#4}{\node[above=2pt] at ($(#2,#1)!0.5!(#3,#1)$) {#4};}
  \IfNoValueF{#6}{\node[below] at (#2,#1-0.12) {#6};}
  \IfNoValueF{#7}{\node[below] at (#3,#1-0.12) {#7};}
}

\makeatletter
\def\BState{\State\hskip-\ALG@thistlm}
\makeatother
\title{Revisiting Real-Time Interval and Throughput Maximization}
\author{Allan Borodin, Changdao He, Nadim Mottu}
\date{\today}

\begin{document}

\maketitle

\begin{abstract}
    Job throughput maximization is the central maximization problem in scheduling. Interval scheduling is the special case of throughput maximization when jobs are intervals and therefore there is no {\it slack} available in which to schedule a job. It is interesting to know to what extent results for interval scheduling can be extended to the more general throughput problem in the real-time model. For the unweighted and proportionally weighted throughput problem (where the weight or value $w_i$ of a job $J_i$ is its processing time $p_i$), there are constant competitive real-time scheduling algorithms using {\it preemption with restarting}. More generally, the result for proportionally weighted interval scheduling can be extended to C-Benevolent weight functions.  We also introduce a new real-time model in which jobs are announced before the actual release time of a job. We show that with sufficient {\it advance notice}, we can obtain a constant competitive ratio for proportionally weighted throughput {\it without any preemption}. 
    However, this advance notice result does not extend to arbitrary C-Benevolent and D-Benevolent weight functions. Finally, we show that unlike interval scheduling, unweighted throughput using {\it preemption with revoking} admits no constant competitive ratio when the number of distinct processing times is unrestricted. More precisely, for instances with at most $k$ distinct processing times, we give a lower bound of $1/(k+1)$ and a deterministic $1/(2k)$-competitive algorithm.
\end{abstract}

% \begin{abstract}
%     Job throughput maximization is the central maximization problem in scheduling. Interval scheduling is the special case of throughput maximization when jobs are intervals and therefore there is no {\it slack} available in which to schedule a job. It is interesting to know to what extent results for interval scheduling can be extended to the more general throughput problem in the real-time model. We show that unlike interval scheduling, unweighted throughput using {\it preemption with revoking} admits no constant competitive ratio when the number of distinct processing times is unrestricted. More precisely, for instances with at most $k$ distinct processing times, we give a lower bound of $1/(k+1)$ and a deterministic $1/(2k)$-competitive algorithm. On the other hand, for the unweighted and proportionally weighted throughput problem (where the weight or value $w_i$ of a job $J_i$ is its processing time $p_i$), there are constant competitive real-time scheduling algorithms using {\it preemption with restarting}. More generally, the result for proportionally weighted interval scheduling can be extended to C-Benevolent weight functions.  We also introduce a new real-time model in which jobs are announced before the actual release time of a job. We show that with sufficient {\it advance notice}, we can obtain a constant competitive ratio for proportionally weighted throughput {\it without any preemption}. 
%     However, this advance notice result does not extend to arbitrary  C-Benevolent and D-Benevolent weight functions.  
% \end{abstract}

\newpage
\section{Introduction}
\label{sec:intro} 

The Weighted Throughput Maximization problem is  a classic problem in scheduling and arguably it is the central maximization scheduling problem. It has been well studied in different online and offline computational models. We will restrict attention to the case of single machine scheduling noting that results for single machine maximization problems can usually be extended to multiple machines, sometimes improving the competitive ratio.  The basic clairvoyant hard deadline throughput maximization problem is defined as follows. 
A job $J_i$ is defined by a 4-tuple: $(r_i,p_i,d_i,w_i)$ where $r(J_i)=r_i \leq r_{i+1}$ is the release time (i.e, the time the job can start executing), $p(J_i)=p_i$ is the processing time, $d(J_i)=d_i$ is the deadline (i.e., the time by which the job must complete), and $w(J_i)=w_i$ is the weight or value  of a job if scheduled and completes by its deadline; that is, a job can possibly be scheduled at any time starting at time $t_i$ if $r_i \leq t_i$ and $t_i + p_i \leq d_i$. A feasible schedule is one in which the scheduled jobs do not intersect. Given a set of jobs ${\cal J}$, the objective is to create a feasible schedule subset ${\cal C} \subseteq {\cal J}$ so as to maximize the $\sum_{J_i \in {\cal C}} w_i$. % COMMENT: Do we want to allow jobs to be scheduled at the deadline of another job?

The slack $s_i$ of a job $J_i$ is defined as $d_i - (r_i + p_i) \geq 0$. 
An important and also well studied special case of Throughput Maximization is interval selection where each job is an interval; that is,  $r_i + p_i = d_i$ or equivalently the slack $s_i = 0$ for all jobs. For one machine, interval scheduling is a selection problem as there is no choice as to where to schedule an interval. 

We will focus on the problem of deterministically scheduling jobs in the real-time model where the scheduling algorithm only learns about a job $J_i$ at time $r_i$.  We note that in the scheduling literature, algorithms in the real-time model are often referred to as online algorithms. We regard online algorithms as those where the $i^{th}$ job arrives at step $i$ without the constraint that $r_i \leq r_{i+1}$. In the real-time model, we can schedule a job at any time $t_i \geq r_i$ and we are usually allowed some form of preemption where the execution of a job can be interrupted any time. In scheduling theory, preemption usually means {\it preemption with resumption} where an interrupted job can be resumed; that is, if $J_i$ has already executed for $e_i < p_i$ time when interrupted, it only needs to execute for an additional $p_i-e_i$ time units. This is in contrast to the { \it preemption with  restarting model} where to reschedule an interrupted job, the job must start executing from its beginning; that is, it must still execute for the full time $p_i$. In the {\it preemption with revoking model}, a job may be aborted at any time but then it is permanently lost. The revoke model is usually studied in the online model and not in the real-time model. 

Given how central and well-studied these problems are, it is surprising that there are still interesting questions to be asked. We will focus on deterministic algorithms in the real-time model with restarting. In the next section we will present previous results that are most relevant to our work in this paper. Here are our main  results. 

\begin{itemize}
    \item For proportional weights (where $w_i = p_i)$, we extend the $\frac{1}{4}$-competitive algorithm from Woeginger \cite{Woeginger94} for interval selection to derive and analyze a $\frac{1}{5}$-competitive deterministic algorithm for the throughput problem with proportional weights in the real-time model with restarting. More generally, as in Woeginger's \cite{Woeginger94} result for interval selection, this algorithm and its analysis applies to all C-Benevolent and D-Benevolent weight functions.  
    
    \item We introduce a new real-time model where the algorithm is given advance notice proportional to the processing time that allows us to derive constant competitive ratios for the throughput problem with proportional weights in the real-time model {\it without preemption}. This result does not extend to arbitrary C-Benevolent and D-Benevolent weight functions. 
    
    \item In contrast to interval selection, for the unweighted throughput problem with preemption and revoking, no deterministic algorithm can obtain a constant competitive ratio when the number of distinct processing times is unrestricted. More precisely, for instances with at most $k$ distinct processing times, we prove a lower bound of $\frac{1}{k+1}$ and give a deterministic $\frac{1}{2k}$-competitive algorithm. This contrasts with the real-time model with restarting, where Hoogeveen, Potts and Woeginger~\cite{HoogeveenPW00} give a deterministic $\frac{1}{2}$-competitive algorithm for unweighted throughput.
\end{itemize}

    We will also discuss the arbitrarily weighted and unweighted throughput problems for the special case of single processing time instances (i.e., $p_i=p$ for some fixed $p$ for all jobs $J_i$) and extensions to input instances with at most $k$ distinct processing times.

% \end{section}
\section{Related Work}
\label{sec:related}
 Although our work only concerns deterministic algorithms in the real-time model, we will also mention analogous results for the online model and for randomized algorithms.  

{\bf Note: For consistency we will state all competitive ratios as fractions $\leq 1$}. Our positive results are strict competitive ratios; that is, we will show $\frac{ALG}{OPT} \geq \rho$ for all input  instances to claim that an algorithm is $\rho$-competitive. All of our negative (impossibility) results can be made asymptotic by adding disjoint copies of a nemesis construction.

{\bf Interval Selection}

Online Unweighted Interval selection (when the $i^{th}$ job arrives at step $i$ unrelated to the release or starting time of the interval) has been considered as a  special case of call admission. It is essentially the call admission problem for a line graph. In the online model, preemption with restarting or resumption doesn't make sense but revoking does make sense.  In the online model, Garey et al \cite{GarayGKMY97} show that there is no deterministic online algorithm with revoking that can achieve a constant competitive ratio for {\it unweighted} interval selection.  They show that no algorithm can have a competitive ratio better than $\frac{1}{\Omega(\log N)}$ in terms of call admission for a $N$ node line graph. Their nemesis contruction can be restated as a $\frac{2}{n}$ inapproximation where $n$ is the length of the input sequence. Borodin and Karavasilis \cite{BorodinK23}   considered the online interval selection problem for instances in which there are at most $k$ distinct interval lengths. They showed that $\frac{1}{2k}$ is the optimal deterministic competitive ratio for such instances in the online model with revoking. For randomized online algorithms, Emek, Halld{\'{o}}rsson and Ros{\'{e}}n \cite{EmekHR2016} provide a $\frac{1}{6}$-competitive algorithm for unweighted interval selection. In contrast, in the real-time model, Woeginger \cite{Woeginger94} shows that the optimal offline greedy algorithm (i.e., ordering by non-decreasing finishing times) of Faigle of Nawijn \cite{FaigleN91} can easily be converted to be an optimal real-time algorithm with revoking. 

Turning to Weighted Interval Selection, Lipton and Tomkins~\cite{Lipton} considered the real-time model without preemption. They analyzed the case of two distinct processing times where jobs have weights proportional to their lengths, and designed a \(1/2\)-competitive randomized algorithm and showed that no algorithm can achieve a better competitive ratio in that model.

Woeginger \cite{Woeginger94} shows that there does not exist a deterministic  real-time algorithm (with revoking) when the weight of an interval is $w(I) = |I|^d$ for $0 < d < 1$ where $|I|$ is the length of the interval $I$. We note that the length of a interval is the processing time when thinking of intervals as a special case of throughput maximization and for interval selection revoking is equivalent to restarting.  
Woeginger \cite{Woeginger94} defined the class of C-Benevolent and D-Benevolent weight functions.

C-Benevolent weight functions \(w(I) = f(|I|)\)  are strictly increasing convex functions \(f\) for \(|I|> 0\) and \(f(0) = 0\). 
D-Benevolent weight functions \(w(I) = f(|I|)\) are non-increasing functions for \(|I| > 0\) and \(f(0)=0\). We  note that C-Benevolent functions include proportional weights \(w(I) = c \cdot|I|\) and D-Benevolent functions include the unweighted case \(w_i = c\) (for any constant \(c\)).

Woeginger considered deterministic algorithms in the real-time model (with revoking) and showed that there is an algorithm that achieves a $\frac{1}{4}$-competitive ratio for both C-Benevolent and D-Benevolent weights. He also showed that this is the optimal ratio for all C-Benevolent weight functions and that there is a D-Benevolent weight function for which $\frac{1}{3}$ is the best possible ratio.  He also showed that the same algorithm provides a $\frac{1}{4}$ competitive ratio for arbitrarily weighted instances in which all intervals have the same length. Woeginger's algorithm is the motivation and starting point for our throughput maximization result for C-Benevolent and D-Benevolent instances. With regard to randomized real-time algorithms (with revoking), Canetti and Irani \cite{CanettiI98} prove non constant lower bounds for arbitrarily weighted instances, where the lower bound is in terms of the ratio between the largest and smallest weights. Fung, Poon, and Zheng \cite{fung_improved_2014} present 1-bit barely random real time algorithms with competitive ratio $\frac{1}{2}$ for arbitrarily weighted single length instances\footnote{Fung et al state all their results for the real-time model but the single instance result holds for the online model. The result for single length instances extends to a $\frac{1}{2k}$ ratio for instances with at most $k$ distinct lengths using $\log k$ random bits by the usual classify and randomly select procedure.}, monotone instances, and C-Benevolent and D-Benevolent weighted instances. Epstein and Levin \cite{EpsteinL10} provide  1.693 randomized lower bounds for equal length and C-Benevolent instances and a 1.5 lower bound for D-Benevolent instances assuming the weight function is surjective. All these results are for the real-time model with revoking.

{\bf Throughput Maximization} 
The problem of optimally maximizing the number of on-time jobs ( equivalently, optimally minimizing the number of tardy jobs), originates in Lenstra et al.~\cite{Lenstra77}, who showed that the problem is NP-hard in the non-preemptive case. This NP-harndess also applies to the preemptive-restart model since in that model an \(OPT\) solution will not use restarting. Lawler~\cite{Lawler90} gave the first polynomial-time algorithm for the preemption-resume model using dynamic programming, and Baptiste~\cite{Baptiste99} later improved the running time to $O(n^4)$, where $n$ is the number of jobs in the input.

Although it is possible to have optimal online algorithms for $NP$-hard problems (since there are no time complexity constraints), we do not expect to obtain optimum online or real-time algorithms for the throughput problem.  To our knowledge, throughput maximization (beyond interval scheduling)  is not studied in the online setting. In addition to the general throughput problem, the problem has been studied for the special case of instances having a single processing time. The single processing time problem with restarting is related to the Quality of Service (QoS) problem for packet routing. The QoS problem can be seen as an integral version of throughput maximization  with unit processing time for all packets. See Vesely et al \cite{VeselyCJS22} for the optimal deterministic algorithm with competitive ratio $\frac{1}{\phi} \approx .618$ where $\phi$ is the golden ratio. The current best randomized algorithm with ratio $1-\frac{1}{e}$ is due to Chin et al \cite{ChinFGHHJLTZZZ15}.  

Hoogeveen, Potts and Woeginer \cite{HoogeveenPW00} adapt the offline Earliest Completion Time algorithm to derive a $\frac{1}{2}$-competitive real-time algorithm with restarting and prove that this is the optimal competitive ratio in this model. Chrobak et al \cite{chrobak_online_2007} study the unweighted case where instances have a single processing time. For this special case, they provide a deterministic algorithm with competitive ratio $\frac{2}{3}$ and prove that this is the optimal competitive ratio. With regard to randomized real-time algorithms for the unweighted throughput problem, Chrobak et al paper gives a 1-bit $\frac{3}{5}$ competitive barely random algorithm for the unweighted single processing time case. Perhaps surprisingly, for the real-time model {\it with preemption-resumption}, Baruha, Haritsa and Sharma \cite{BaruhaHS94} show that no deterministic algorithm can achieve a constant ratio for the unweighted throughput problem. We note that although  every real-time algorithm with restarting is simulated by a real-time algorithm with resuming, an optimal offline algorithm can also benefit by resuming, whereas the optimal offline algorithm with restarting never needs to restart. However, if the slack $s_i = d_i-r_i - p_i$ is sufficiently large, Lucier et al \cite{LucierMNY13} provide a constant competitive deterministic algorithm for the unweighted throughput problem.  While Kalyanasundaram and Pruhs \cite{KalyanasundaramP03} provide a constant $c$-competitive barely random algorithm  for the unweighted throughput problem, the constant $c$ is (remarkably) impractical.    

For the weighted throughput problem, the impossibility results for the weighted interval selection problem immediately show that there cannot be a constant competitive real-time algorithm (with restarting) when the weights are arbitrary. Goldman, Parwatikar and Suri \cite{goldman_online_2000} consider the proportionally weighted throughput problem for instances having a single processing time and instances having at most two distinct processing times. This show that without any preemption, $\frac{1}{2}$ is the optimal deterministic ratio for the single processing time case (which is equivalent to the unweighted case when there is a single processing time), and provide a $\frac{1}{4}$ randomized competitive ratio for the case of two distinct processing times.  To the best of our knowledge, the throughput problem has not been previously studied for real-time algorithms (with restarting) for restricted classes of weighted instances (beyond single length instances). 
Fung, Poon and Zhang \cite{fung_improved_2014} give a barely random $\frac{1}{3}$ competitive algorithm for arbitrarily weighted instances with a single processing time. As in the weighted interval selection problem, this result extends to $k$ distinct processing times to obtain a $\frac{1}{3k}$ randomized competitive ratio. In the real-time model with resumption, Koren and Shasha \cite{KorenS92} provide a deterministic algorithm with competitive ratio $\frac{1}{(1+\sqrt{r})}$ for $r$-proportional weight functions that satisfy 
$p_i \leq w_i \leq r \cdot  p_i$ for every job $J_i$. This matches the deterministic inapproximation in Baruah et al. \cite{BaruahKMMRRSW92}. By making $r$ a function of the processing time, we obtain a  non-constant inapproximation for this class of weight functions.

\section{C-Benevolent Throughput with Restarting}
    \label{sec:proportional weights}
    In this section we concern ourselves with input sequences in which the value of a job is a C-Benevolent function as defined in the previous section. When the slack is zero and revoking is permitted, the work of Woeginger \cite{Woeginger94} gives us a tight competitive ratio for interval selection. That is to say, he provides an algorithm that is \(\frac{1}{4}\)-competitive for C-benevolent input sequences, and proves that no algorithm can do better. We show that we can extend his algorithm for arbitrary slack by using restarting instead of revoking.   In this case, the resulting algorithm is \(\frac{1}{5}\)-competitive.
% In Theorem \ref{thm:k-length-revoke-lower} we show that revoking is not sufficient. 

We define the following algorithm, we call the \(\tau\)-\textsc{Persist} algorithm: 

\begin{algorithm}[H]
    \caption{Code for \(\tau\)-\textsc{Persist} where \(\tau > 1\)}
    \label{alg:tau-persist}
    \begin{algorithmic}[1]    
        \State When a job \(J_k\) is announced at time \(T\):
        \State \(J_c \leftarrow\) job currently being processed or \(\emptyset\) if there are no such jobs
        \State \(\sigma_c \leftarrow\) the time \(J_c\) was started if it is not \(\emptyset\). 
        \If {\(J_c = \emptyset \)}
        \State Start processing \(J_k\)
        \ElsIf{\(w(J_k) > \tau \cdot w(J_c)\) or (\(T+p_k < \sigma_c + p_c\) and \(w(J_k) \geq w(J_c)\))}\label{alg_line:tau-persist-interrupt-condition}
        \State Interrupt \(J_c\)
        \State Start processing \(J_k\)
        \EndIf
        \State When a job \(J_c\) has completed
        \State \(Q\leftarrow \) the set of all jobs that have not expired and have not been completed
        \State \(J_k \leftarrow\) the job with maximum weight in \(Q\)
        \State Start processing \(J_k\)
    \end{algorithmic}
\end{algorithm}

In other words, this algorithm processes the largest job it can, but only interrupts a job when it can be replaced with a job having \(\tau\) times more weight or if it can be completed earlier than the current job while having at least as much weight\footnote{It will never be the case that a job can be completed earlier than the current job while having at least as much weight when the input sequence is C-benevolent. With D-benevolent instances, however, this may happen.}. We note that this algorithm uses \emph{restarts}, meaning that if a job is interrupted, it is added to the set of jobs that can be scheduled. This means that it can be scheduled after another job is completed.

We now analyze the performance of \(\tau\)-\textsc{Persist}. This analysis borrows ideas from Woeginger's analysis of the interval selection algorithm he calls \(\textsc{Heu}\) \cite{Woeginger94}. The algorithm \(\textsc{Heu}\) acts identically to how \(2\)-\textsc{Persist} acts on input sequences where slack is zero.

We begin with some definitions that will be useful for our charging arguments. We define the range of a job \(J_i\) as follows. Suppose that \(J_i\) was scheduled at a point in which the machine was idle and thus did not interrupt another job, then its \emph{predecessor chain} is an empty list. Otherwise, suppose that \(J_i\) was scheduled interrupting \(J_c\). Then the predecessor chain of \(J_i\) is the predecessor chain of \(J_c\), along with the tuple \((J_c, \sigma'_c)\) appended to the beginning of the list, where \(\sigma'_c\) is the value of start time of \(J_c\) before it was interrupted. The predecessor chain of a job is defined adaptively as the algorithm executes. When a job is interrupted, its predecessor chain then becomes undefined until it is scheduled again later. This means that when the algorithm has finished executing, only jobs completed by the algorithm have a predecessor chain. The \emph{successor}, of \(J_i\) is the highest processing time job that can be scheduled during the time our algorithm is processing \(J_i\). If nothing intersects it, then the successor is treated as a job of length and weight equal to \(0\). Suppose that the predecessor chain of \(J_i\) is \((\sigma_1', J_1'), \ldots, (\sigma_h', J_h')\) and the successor of \(J_i\) is \(J_c\). We define the \emph{range} of \(J_i\) as the interval \([\sigma_h', \sigma_i + p_i+ p_c]\).

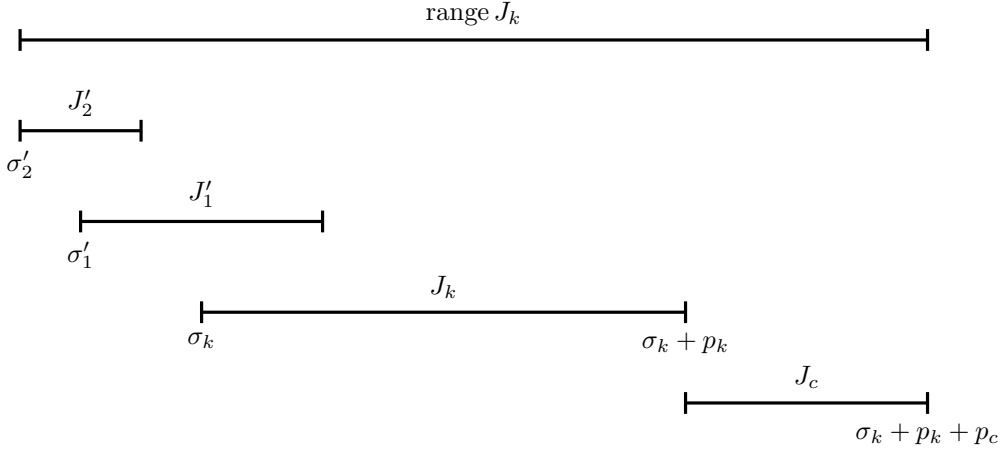
\begin{figure}[ht]
    \centering
    \begin{tikzpicture}[x=0.8cm,y=1.2cm,>={Stealth[length=2.5pt]}]
    % axis

    \interval{0}{15}[\(\range{J_k}\)][black][][]
    \interval[-1]{0}{2}[\(J_2'\)][black][\(\sigma_2'\)][]
    \interval[-2]{1}{5}[\(J_1'\)][black][\(\sigma_1'\)][]
    \interval[-3]{3}{11}[\(J_k\)][black][\(\sigma_k\)][\(\sigma_k +p_k\)]
    \interval[-4]{11}{15}[\(J_c\)][black][][\(\sigma_k +p_k+p_c\)]

    % \interval{1}{19}[][black][$r_1$][$d_1$]
    % \interval{6}{11.5}[][black][$a_1$][$b_1$]
    % \interval[-1]{6.5}{11.5}[][black][$r_2$][$d_2$]
    % \interval[-1]{10}{11.5}[][black][$r_2+s_2$]
    % \intervalthin[-2]{10}{10.5}[$\varepsilon$][black]
    % \interval[-2]{10.5}{11.5}[][black][$r_4$][$d_4$]
    \end{tikzpicture}
    \caption{The range of a job \(J_k\)}
    \label{fig:points}
\end{figure}

\medskip
\begin{observation}\label{obs:ranges}
    Suppose a job \(J_i\), not completed by \(\tau\)-\textsc{Persist} can be scheduled a time \(\sigma_i'\). Then \([\sigma_i', \sigma_i' + p_i)\) is contained in the range of some completed job.
\end{observation}

\begin{proof}
    Suppose for the sake of obtaining a contradiction that there is some job \(J_i\) which can be scheduled at time \(\sigma_i'\) such that \(J_i\) has not been completed elsewhere by the algorithm and that is not contained in the range of some job.
    
    The time \(\sigma_i'\) occurs after the start of some range since \(\tau\)-\textsc{Persist} always schedules jobs that are released if it is not currently processing a job. Let \(r_1 = \range(J_k)\) be the last range before the time \(\sigma_i'\)
    
    \begin{description}
    \item[Case 1:] \([\sigma_i', \sigma_i' + p_i)\) does not intersect \([\sigma_k, \sigma_k+p_k)\).
    Since \([\sigma_i', \sigma_i' + p_i)\) is not contained in the range of any job, no job is being scheduled at \(\sigma_i'\), otherwise \([\sigma_i', \sigma_i' + p_i)\) would be contained in the predecessor chain of some job. This means our algorithm should schedule \(J_i\) since it hasn't completed it. This contradicts the fact that \(J_i\) is not in the predecessor chain of \(J_h\).
    \item[Case 2:] \([\sigma_i', \sigma_i' + p_i)\) intersects \([\sigma_k, \sigma_k+p_k)\).
    In this case, \(J_i\) can be scheduled during the time \(J_k\) is being processed and thus has a processing time less than or equal to the processing time of the successor \(J_c\) of \(J_k\). Since \(\sigma_k + p_k + p_c\) is the latest that any job that can be scheduled during \(J_k\) can finish, \(J_i\) must complete before the end of the range, contradicting our assumption.
    \end{description}
\end{proof}

We will use Karamata's inequality \cite{karamata} in our analysis of $\tau$-Persist:
\medskip
\begin{lemma}[Karamata's inequality]
    Suppose that some function \(f\) is C-benevolent. Suppose \(x_1 \leq x_2 \leq \ldots \leq x_n\) and \(y_1 \leq y_2 \leq \ldots y_n\). Furthermore, suppose that \(\sum_{i=1}^j x_i \leq \sum_{i=1}^j y_i\) for all \(0<j \leq n\). Then \(\sum_{i=1}^n f(x_i) \leq \sum_{i=1}^n f(y_i)\).
\end{lemma}
\medskip
\begin{lemma}\label{lem:prop-charging}
    Suppose that we run \(\tau\)-\textsc{Persist} against a C-benevolent input sequence. Consider some job \(J_k\) completed by our algorithm. Any schedule that is contained in \(\range(J_k)\) which consists only of jobs that have not been completed by 
    \(\tau\)-\textsc{Persist}, that are not contained in the range of a later job, must have a total weight of less \(\frac{\tau^2}{\tau - 1} \cdot w_k\).
\end{lemma}

\begin{proof}
    Let \((\sigma_1', J_1'),\ldots, (\sigma_h', J_h')\) be the predecessor chain of a job \(J_k\), sorted from latest to earliest scheduled (in other words \(J_h'\) is the earliest job scheduled and \(J_1'\) is the latest) and let \(J_c\) be its successor.

    Let \(y_1, \ldots, y_\ell\) be the list \(p(J_1'),\ldots, p(J_h'), p(J_k)\) sorted in non-increasing order where \(p(J)\) is the processing time of job \(J\).

    Consider some schedule by \(OPT\) consisting of jobs, \(G_1,\ldots,G_{q}\), each scheduled at the times \(\sigma_1^*,\ldots, \sigma_q^*\) respectively. Suppose that \(\tau\)-\textsc{Persist} has not completed any of the jobs \(G_1,\ldots,G_q\) and that the schedule constructed with these jobs is contained in \(\range(J_k)\). Without loss of generality, we assume that \(OPT\) schedules one job that intersects \(J_k\) on the right, let \(G_q\) be that job. Let \(x_1, \ldots, x_\ell\) be the list \(p(G_1),\ldots, p(G_{q-1})\) sorted in non-increasing order. We pad the shorter of the lists \(y_1\ldots, y_\ell\) and \(x_1 \ldots, x_\ell\) with \(0\) to ensure they are the same length.

    We will show that for every \(0<j\leq \ell\) we have that \(\sum^{j}_{i=1} x_i\leq \sum^{j}_{i=1} y_i\).

    Suppose for the sake of obtaining a contradiction, that this is not the case, and that there is some \(j\) such that \(\sum^{j}_{i=1} x_i> \sum^{j}_{i=1} y_i\). Let \(t\) be the smallest such \(t\). Then consider \(y_t\), this length is no larger than \(y_j\) for \(j \leq t-1\) other jobs. 
    
    Suppose that \(y_t = p(J_u')\) for some job \(J_u'\) either in the predecessor chain of \(J_t\) or \(J_u' = J_k\). Since \(t\) is minimal, \(x_t\) must be greater than \(y_t\), otherwise the sum of the \(x_i\) would not overtake the sum of the \(y_i\). Then \(OPT\) scheduled some job \(G_r\) with value greater than \(p(J_u')\). Such a job would need to be scheduled after \(J_u'\) was scheduled by \(\tau\)-\(\textsc{Persist}\), otherwise \(J_u'\) would not be scheduled over that job. But note that any job that is larger than it would also need to be scheduled after \(J_u'\), which means that every processing time in \(x_1 \ldots x_t\) must fit in the interval \([\sigma_u', \sigma_k + p_k]\). Note that the values of \(y_i\), \(j \leq t\) are the lengths of jobs which have not been completed by \(\tau\)-\(\textsc{Persist}\) and thus the total length of \([\sigma_u', \sigma_k + p_k]\) is smaller than the sum of all \(y_i\), \(j \leq t\). This is a contradiction.

    If \(y_t \neq p(J_{u}')\) for any \(J_u'\) then \(y_t\) is a \(0\) added for padding. Then this means that all lengths are accounted for in the sum \(\sum^{t}_{i=1} y_i\). Since all the \(x_i\) must fit within \([\sigma_h', \sigma_k+p_k)\), their sum must be smaller than the sum of all lengths in the predecessor chain along with \(J_k\).

    % \textbf{Base Case:} The value of \(y_1\) is the largest processing time job processed by the algorithm or the successor of \(J_k\). This must be \(J_k\) since if it were in the predecessor chain it would be not be interrupted. The highest processing time job scheduled by \(OPT\) cannot have processing time higher than this without intersecting with \(J_k\) on the right, otherwise \(J_k\) would not be scheduled.

    % \textbf{Induction Case:} Suppose that for every \(1<j\leq t-1\) we have that \(\sum^{j}_{i=1} x_i\leq \sum^{j}_{i=1} y_i\).

    This means that by Karamata's inequality \(w(G_1) + \cdots + w(G_{q-1}) =\sum^\ell_{i=0} f(x_i) \leq \sum^\ell_{i=0} f(y_i) = w(J_1') + \cdots + w(J_h') + w(J_k)\). Furthermore, since the successor of \(J_k\) is by definition the largest job that can be scheduled intersecting \(J_k\) on the right, \(w(J_c) \geq w(G_q)\).

    We also know that since a job is only interrupted when it is \(\tau\) times bigger than the previous job, we can conclude that the \(i^\text{th}\) job in the predecessor chain is has weight at most \(\tau^{-i}\cdot w_k\). Furthermore, the successor of \(J_k\) has weight at most \(\tau \cdot w_k\). Putting this all together we have:

    \[\sum_{i=1}^q w(G_i) \leq \sum_{i=1}^h [w(J_i')] + w(J_k) + w(J_c) \leq \sum^\infty_{i=1} \tau^{2-i}w_k = \frac{\tau^2}{\tau - 1} \cdot w_k\]

\end{proof}

\begin{corollary} \label{cor:C-benv-tau-persist}
    \(\tau\)-\textsc{Persist} is \(\frac{\tau - 1}{\tau^2 + \tau - 1}\)-competitive against instances that are C-benevolent. In particular, this function is maximized at \(\tau = 2\) giving a \(\frac{1}{5}\)-competitive algorithm.
\end{corollary}
\begin{proof}
    Suppose for the sake of obtaining a contradiction, that there exists some input sequence \(J_1,\ldots, J_n\) such that our algorithm gets the total value \(v\), and there exists a schedule that obtains a value \(v' > \frac{\tau^2+\tau -1}{\tau-1}\cdot v\). Consider this schedule with all the jobs completed by our algorithm removed. By observation \ref{obs:ranges}, all jobs in this schedule have to be scheduled in the range of some job completed by the algorithm. By lemma \ref{lem:prop-charging}, this means that the total value of this schedule \(v'' \leq \frac{\tau^2}{\tau-1}\cdot v\). Since the most value that could have been lost by removing all the jobs completed by the algorithm is \(v\), \(v' \leq v'' + v \leq (\frac{\tau^2}{\tau-1} +1)\cdot v=\frac{\tau^2+\tau -1}{\tau-1}\cdot v\) which is a contradiction.
\end{proof}

% \begin{theorem} \label{thm:1/5-comp-c-benevolent}
%     There exists a \(\frac{1}{5}\)-competitive algorithm for real-time throughput maximization using restarting against C-benevolent input sequences and in particular for proportional weights.
% \end{theorem}
% \begin{proof}
%     For \(\tau > 1\) the maximum of \(\frac{\tau - 1}{\tau^2 + \tau - 1}\) is \(\frac{1}{5}\) which is achieved at \(\tau = 2\). This means that the algorithm \(2\)-\textsc{Persist} is \(\frac{1}{5}\) competitive as desired.
% \end{proof}

\begin{remark} % TODO rewrite this and talk more about the intuition that J_k cannot be scheduled outside the range for interval selection.
    In the case of the real-time interval selection problem, where slack is 0, we note that schedules for the same input sequence only differ in their choice of which jobs to schedule and not by the times at which they are scheduled. With interval selection, if a job \(J_k\) is scheduled by both \(\tau\)-\textsc{Persist} and \(OPT\), then \(OPT\) schedules it at the same time as \(\tau\)-\textsc{Persist} in \(\range(J_k)\). As such applying lemma \ref{lem:prop-charging} tells us that \(\tau\)-\textsc{Persist} is \(\frac{\tau - 1}{\tau^2}\)-competitive against such input sequences. This matches the fact that as previously noted, Woeginger's algorithm \cite{Woeginger94} is \(\frac{1}{4}\)-competitive and acts identically to \(2\)-\textsc{Persist}. We can therefore think of \cref{lem:prop-charging} as an alternative proof of Woeginger's algorithm's correctness.
\end{remark}

\section{D-Benevolent Throughput with Restarting}
    \label{sec:D-Benevolent}
    We now show that \(\tau\)-\textsc{Persist} also performs well against instances that are D-benevolent.

A function \(f\colon \reals_{\geq 0} \to \reals\) is D-Benevolent \cite{Woeginger94} if it satisfies the following properties:
\begin{enumerate}
    \item \(f(0) = 0\) and \(x > 0 \implies f(x) > 0\).
    \item \(0< x_1 < x_2 \implies f(x_1) \geq f(x_2)\). 
\end{enumerate}

We begin by proving the analogue to lemma \ref{lem:prop-charging}.

\medskip
\begin{lemma}\label{lem:d-benevolent-charging}
    Suppose that we run \(\tau\)-\textsc{Persist} against a D-benevolent input sequence. Any schedule that is contained in \(\range(J_k)\) and consists only of jobs that are neither contained in the range of a job scheduled later or completed by \(\tau\)-\textsc{Persist} must have a total weight of less than \(\frac{\tau^2}{\tau - 1}\cdot w_k\).
\end{lemma}

\begin{proof}
    Let \((\sigma_1', J_1'),\ldots, (\sigma_h', J_h')\) be the predecessor chain of a job \(J_k\), sorted from latest to earliest scheduled (in other words \(J_h'\) is the earliest job scheduled and \(J_1'\) is the latest) and let \(J_c\) be its successor. Consider the set of points \(a_{0} = \sigma_k+p_k+p_c, a_{1} = \sigma_k+p_k, a_2 = \sigma_k, a_{3} =\sigma_1',\ldots, a_{h+2} = \sigma_h' \).

    By definition of the algorithm, a job is only interrupted when the job that replaces it has \(\tau\) times its weight. Therefore, \(w_\ell' \leq \tau\cdot w_{\ell+1}' \leq \tau^j \cdot w_{\ell+j}\). This implies that if a job covers the point \(a_j\), its weight is at most \(\tau^{2-j}w_k\) for any \(1 \leq j \leq h\) unless the job was previously completed by the algorithm. For \(j \geq 2\), this follows from the fact that \(a_j\) is the start time of job \(J_{j+2}'\) and thus if a job with weight higher that \(\tau\cdot w_{j+2}\leq \tau^{2-j}w_k\) existed at that time, then it would be scheduled instead. For \(j = 1\) it follows from the fact that \(J_k\) was not interrupted, and therefore no job with weight \(\geq \tau\cdot w_k\) could be scheduled during its execution.

    Consider some schedule \((\sigma_0^*, G_0),\ldots, (\sigma_q^*, G_q)\) consisting of jobs \(G_0,\ldots,G_q\) that the algorithm has not scheduled that is contained in \(\range(J_k)\) and such that no \([\sigma_i^*, \sigma_i^* + p(G_i))\) is contained in \(\range(J_r)\) for some \(J_r\) scheduled after \(J_k\) by \(\tau\)-\textsc{Persist}.

    We claim that for every \(0\leq i\leq q\) \([\sigma_i^*, \sigma_i^* + p(G_i))\) overlaps with \(a_j\) for some \(j\). Suppose that some \(G_\ell\) does not overlap with any \(a_j\). Then \([\sigma_\ell^*,\sigma_\ell^* + p(G_\ell)) \subseteq (a_{u+1}, a_{u})\) for some \(0 \leq u < q\). If \(u > 0\) then by line \ref{alg_line:tau-persist-interrupt-condition} of algorithm \ref{alg:tau-persist}, \(G_\ell\) would interrupt the current job since it would be completed sooner and would have a larger weight since the input sequence is D-benevolent. This means that \(G_\ell\) is in the predecessor chain of \(J_k\). Since the input sequence is D-benevolent, jobs that appear later in the predecessor chain have larger weights and have larger weights. Since \(G_\ell\) has processing time smaller than the job whose starting point is at \(a_{u+1}\), after all, it can complete before this job is interrupted), it must be released before said job which is a contradiction. If \(u = 0\), then \(G_\ell\) is scheduled after \(J_k\) is completed. This means that the set of jobs some job that \(\tau\)-\textsc{Persist} has not completed could be scheduled as soon as \(J_k\) completed and thus \([\sigma_\ell^*, \sigma_\ell^* + p(G_\ell))\) is contained in the range of a later job.

    This means that every job in the schedule overlaps with at least one point in \(a_0, \ldots, a_{h+2}\). This means that the total weight of this schedule is less than \(\sum^\infty_{i=1} \tau^{2-i}w_k = \frac{\tau^2}{\tau - 1} \cdot w_k\), as desired.
\end{proof}

The next result follows exactly from the same argument as \cref{cor:C-benv-tau-persist}.

\medskip
\begin{restatable}{theorem}{dbenevolentrestarting}\label{thm:d-benevolent-restarting}
    \(\tau\)-\textsc{Persist} is \(\frac{\tau - 1}{\tau^2 + \tau - 1}\)-competitive against instances that are D-benevolent. In particular, there exists a \(\frac{1}{5}\)-competitive algorithm for real-time throughput maximization using restarting against D-benevolent input sequences.
\end{restatable}

We also can extend this result to other similar instances where we can ensure that each job intersects with an endpoint of a job that \(\tau\)-\textsc{Persists} schedules.
\medskip
\begin{theorem}
    \(\tau\)-\textsc{Persist} is \(\frac{\tau - 1}{\tau^2 + \tau - 1}\)-competitive against instances where all lengths are the same and the weights are arbitrary. In particular, there exists a \(\frac{1}{5}\)-competitive algorithm for real-time throughput maximization using restarting against such instances.
\end{theorem}
\begin{proof}
    The proof is the same as the proof of \cref{lem:d-benevolent-charging}. We note that since there is only one length for jobs, they must all intersect with a point in \(a_0,\ldots,a_{h+2}\).
\end{proof}

\section{Advance Notice}
    \label{sec:advance-notice}
    %\subsection{What is Advance Notice?}
So far we have looked at algorithms that use preemption (either revoking or restarting) to achieve a constant competitive ratio. In this section, we explore if one can achieve comparable results without preemption if the algorithm knows about jobs in advance.

For this section, we define a job \(J_i = (a_i, r_i, p_i,d_i,w_i)\) as a \(5\)-tuple in \(\reals_{\geq 0}^5\). Here \(p_i\), \(d_i\), and \(w_i\) are unchanged from the standard throughput model. But now, rather than the algorithm being made aware of the job at time \(r_i\), it is made aware of the job at time $a_i \leq r_i$. Even so, the job cannot be scheduled at any time prior to \(r_i\). This allows the algorithm to potentially make decisions about which jobs to schedule with some knowledge about potential future jobs which have not been released yet. 

%In general, the inclusion of an announcement time does not necessarily help us against an adversary relative to a model where jobs are announced at their release times. Since announcing a job \(J_i\) such that \(a_i < r_i\) does not impact the value of the optimal offline solution (which knows about every job from the very beginning), an adversary designing the sequence of jobs $\mathcal{J}$ can always make an algorithm perform worse by making $a_i = r_i$ without impacting the optimal solution. 

\medskip
\begin{definition}\label{t-ad}
    A job \(J\) has \(t\)-advance-notice for some constant \(t\in \mathbb{R}_{\geq 0}\) if \((r_i - a_i) \geq t\cdot p_i\). An input sequence has \(t\)-advance-notice if all jobs in it have \(t\)-advance-notice.
\end{definition}
Since jobs are designed by an adversary and giving an earlier $a_i$ only benefits an online algorithm without changing the optimum, we can assume that $(r_i - a_i) = t\cdot p_i$.

\medskip
\begin{remark}
    If an input sequence has \(t\)-advance-notice it also has \(t'\)-advance-notice for all \(t' \leq t\).
\end{remark}

\subsection{Algorithms that use Advance Notice:} \label{subsec:advance-notice-alg}

We show that given any \(t > 0\), there exists an algorithm which achieves a constant competitive ratio against input sequences with \(t\)-advance-notice and proportional weights without using preemption. 

We begin by extending the algorithm \(\tau\)-\textsc{Persist} as defined in section \ref{sec:proportional weights} to take advantage of earlier announcement times. This algorithm still uses restarting if insufficient advance-notice is given. We will later prove that on input sequences with \(\frac{1}{\tau}\)-advance-notice, the algorithm never uses preemption.
Our augmented algorithm, which we will call \(\tau\)-\(\textsc{Persist}^*\) works as follows, we simply run \(\tau\)-\textsc{Persist} on the set of known jobs and simulate the algorithm into the future, as if no new jobs get announced. When a job would normally be scheduled by \(\tau\)-\textsc{Persist}, but the algorithm knows (from simulating \(\tau\)-\textsc{Persist} into the future) that it will be preempted by a job yet to be released, the algorithm does not schedule the job, although it still interrupts the currently running job if there is one.

\medskip
\begin{lemma}
    The schedule constructed by \(\tau\)-\(\textsc{Persist}^*\) and by \(\tau\)-\(\textsc{Persist}\) are the same if run on the same input sequence.
\end{lemma}
\begin{proof}
    We fix an input sequence \(\mathcal{J}\). 
    
    By definition, \(\tau\)-\(\textsc{Persist}^*\) only ever schedules a job if \(\tau\)-\(\textsc{Persist}\) does. It also interrupts the execution of a job whenever \(\tau\)-\(\textsc{Persist}\) does. As such it suffices to show that every job completed by \(\tau\)-\(\textsc{Persist}\) is completed by \(\tau\)-\(\textsc{Persist}^*\) at the same time. 
    
    Suppose for the sake of obtaining a contradiction that there are some jobs that are completed by \(\tau\)-\(\textsc{Persist}\) that are not completed by \(\tau\)-\(\textsc{Persist}^*\) at the same time. Let \(J_i\) the earliest such job that is completed by \(\tau\)-\(\textsc{Persist}\).

    Let \(\sigma_i\) be the time at which \(\tau\)-\(\textsc{Persist}\) started processing \(J_i\) when it was completed. Since \(\tau\)-\(\textsc{Persist}^*\) did not schedule this job, by definition, there exists some job \(J_j\) such that \(a_j \leq r_i\), \(w_j > \tau\cdot w_i\), and \(r_j \in (\sigma_i, \sigma_i + p_i)\).

    At time \(r_j\), suppose that \(J_i\) is being run by \(\tau\)-\(\textsc{Persist}\). Then by definition of the algorithm, it should interrupt \(J_i\) allow for \(J_j\) to be scheduled. Otherwise, \(J_i\) has since been interrupted by another job, we get a contradiction since \(J_i\) must be completed when scheduled at time \(\sigma_i\).
\end{proof}

By this lemma, we can define the \emph{predecessor chain}, \emph{successor}, and \emph{range} of a job \(J\) being scheduled by \(\tau\)-\(\textsc{Persist}^*\) to be the predecessor chain, successor, and range of this job if it were scheduled by \(\tau\)-\(\textsc{Persist}\). We also define the \emph{predecessor} of \(J\) to be the first element in the predecessor chain of \(J\). Note that the first element in the predecessor is the last job to be interrupted before \(J\) is scheduled since jobs are appended to the start of the predecessor chain.

\medskip
\begin{theorem}\label{thm:no_premption_needed}
    Suppose that some job \(J_k\), in an input sequence with proportional weights, is scheduled by \(\tau\)-\(\textsc{Persist}^*\) (but possibly not completed) and has \(\frac{1}{\tau}\)-advance-notice. Then the predecessor of \(J_k\) was not interrupted by \(\tau\)-\(\textsc{Persist}^*\).
\end{theorem}

\begin{proof}
    If \(J_k\) has no predecessor, then \(J_k\) did not interrupt a job before being scheduled. Otherwise, let\(J_\ell\) be the predecessor of \(J_k\) and let \(\sigma_\ell\) be the time at which \(J_\ell\) was scheduled by \(\tau\)-\(\textsc{Persist}\) before it was interrupted.

    Since \(J_k\) interrupted \(J_\ell\), \(r_k \in (\sigma_\ell,\sigma_\ell+p_\ell)\). Furthermore, since \(J_k\) has \(\frac{1}{\tau}\)-advance-notice, \((r_k - a_k) \geq \tau^{-1}\cdot p_k\). Since the input sequence has proportional weights, we know that \(p_\ell = w_\ell\) and \(p_k = w_k\). And finally, since \(J_k\) interrupted \(J_\ell\), \(p_\ell \cdot \tau < p_k\). Since \(a_k \leq r_k - \tau^{-1}\cdot p_k < \sigma_\ell + p_\ell - \tau^{-1}\cdot p_k <\sigma_\ell + p_\ell - p_\ell =\sigma_\ell\) we get \(a_k< \sigma_\ell\).

    Therefore, since \(a_k < \sigma_\ell\), \(J_k\) was announced before \(J_\ell\) was scheduled, and thus by definition, \(\tau\)-\(\textsc{Persist}^*\) does not schedule \(J_\ell\).
\end{proof}

\begin{corollary}
    For any \(0<t \leq \frac{1}{2}\), there is an algorithm that achieves \(\frac{t-t^2}{1+t-t^2}\)-competitiveness against input sequences with \(t\)-advance-notice and proportional weights without using preemption. Furthermore, for any \(t \geq \frac{1}{2}\), there is an algorithm that achieves \(\frac{1}{5}\)-competitiveness under the same conditions.
\end{corollary}

\medskip
\begin{remark}
    As noted in a previous remark, when slack is zero, we get a slightly better competitive ratio from \(\tau\)-\textsc{Persist}, which becomes \(\frac{\tau -1}{\tau^2}\)-competitive. It follows that this better competitive ratio also carries to the case of advance-notice.
\end{remark}

% 

% We also have the following negative result, which we will prove in appendix \ref{sec:advance-notice-negative}.

% \medskip

\subsection{A Negative Result for Preemption-Free Advance Notice with proportional weights} 
    \label{sec:advance-notice-negative} % are these upper bounds? 
    
\begin{restatable}{theorem}{negativeadvancenotice}
    For every \(t\) and for any throughput maximization algorithm that does not use preemption, there exists an input sequence with proportional weights and \(t\)-advance-notice such that the competitive ratio of the the algorithm against this input sequence is \(\frac{t}{2t+1} - \delta\) for any \(\delta > 0\).
\end{restatable}

\begin{proof}
    We will show that for any $\varepsilon > 0$ and any online algorithm $ALG$ there exists a sequence of jobs $\mathcal{J}$ such that $ALG \leq OPT \cdot \frac{t}{2t+1} + \varepsilon$.

    First we fix $\gamma \leq \frac{\varepsilon (t\cdot (2t+1))}{2+t}$.
    
    First we announce $J_1$ with $a_1 = 0$, $r_1 = t$, $p_1 = 1$ and $d_1$ arbitrarily large.
    
    If $J_1$ is never scheduled by $ALG$ then the adversary declares it to be the last job and we have infinite competitive ratio.
    
    If $J_1$ is scheduled at time $s_1$ then at $s_1 + \gamma$ we announce $J_2$ with $a_2 = s_1+ \gamma$, $r_2 = s_1+ 1 - 2\cdot \gamma$, $p_2 = \frac{p_1 - \gamma}{t}$ and $d_2 = r_2 + p_2$. Note that $J_2$ is impossible to schedule since no preemption is permitted.
    Next the adversary announces $J_3$ with $a_3 = s_1 + \gamma$, $r_3 = s_1+ 2\cdot \gamma$ $p_3 = \frac{\gamma}{t}$ $d_3 = r_3 + p_3$. The adversary will also announce: $J_i$ with $a_i = s_1 + \gamma $, $r_i = d_{i-1}$, $p_i = d_i - r_i$ and $d_i = \min(r_i + \frac{\gamma}{t}, r_2)$. For $i > 3$ until the interval $[s_1, s_1 + p_1]$ is covered.
    
    In this schedule, the only job that the algorithm has done is $J_1$ which has value $p_1$. Meanwhile an optimal algorithm $OPT$ can schedule all jobs since none of them overlap except for $J_1$, and since $J_1$ has arbitrarily large slack, we can always fit it somewhere else.
    
    So the value of $ALG = p_1 = 1$ and $OPT = p_1 + p_2 + p_3 + \sum_i p_i = 1 + \frac{1 - \gamma}{t} + (1 - 2\gamma) = \frac{2t + 1 - 2\gamma - \gamma\cdot t }{t}= \frac{2t + 1}{t} - \frac{2\cdot\gamma + \gamma\cdot t}{t}$.
    
    So $OPT \cdot \frac{t}{2t+1} = ALG - \frac{(2\gamma + \gamma \cdot t)\cdot}{t\cdot (2t+1)}$.  By our choice of $\gamma$ we get that $ALG \leq OPT \cdot \frac{t}{2t+1} + \varepsilon$
\end{proof}
\begin{figure}[ht]
    \centering
    \begin{tikzpicture}[x=0.8cm,y=1.2cm,>={Stealth[length=2.5pt]}]
    % axis
    \node[draw] at (0,0) {\(ALG\)};
    \interval{2}{5}[][dotted][\(a_1\)][]
    \interval{2}{2}[][black][][]
    \interval{5}{5}[][black][\(r_1\)][]
    \interval{20}{20}[][black][][\(d_1\)]
    \interval{7}{10}[\(J_1\)][blue][][]

    \node[draw] at (0,-1) {\(OPT\)};
    \interval[-1]{2}{5}[][dotted][\(a_1\)][]
    \interval[-1]{2}{2}[][black][][]
    \interval[-1]{5}{5}[][black][\(r_1\)][]
    \interval[-1]{20}{20}[][black][][\(d_1\)]
    \interval[-1]{7}{7}[][black][\(a_2,a_3\)][]
    \interval[-1]{7.2}{7.4}[][olive][][]
    \interval[-1]{7.4}{7.6}[][olive][][]
    \interval[-1]{7.6}{7.8}[][olive][][]
    \interval[-1]{7.8}{8.0}[][olive][][]
    \interval[-1]{8.0}{8.2}[][olive][][]
    \interval[-1]{8.2}{8.4}[][olive][][]
    \interval[-1]{8.4}{8.6}[][olive][][]
    \interval[-1]{8.6}{8.8}[][olive][][]
    \interval[-1]{8.8}{9.0}[][olive][][]
    \interval[-1]{9.0}{9.2}[][olive][][]
    \interval[-1]{9.2}{9.4}[][olive][][]
    \interval[-1]{9.4}{9.6}[][olive][][]
    \interval[-1]{9.6}{9.8}[][olive][][]
    \interval[-1]{9.8}{10}[][olive][][]
    \interval[-1]{7.2}{10}[\(J_3,\ldots,J_n\)][olive][][]
    \interval[-1]{10}{13}[\(J_2\)][red][][]
    
    \interval[-1]{14}{17}[\(J_1\)][blue][][]
    
    % \interval{1}{19}[][black][$r_1$][$d_1$]
    % \interval{6}{11.5}[][black][$a_1$][$b_1$]
    % \interval[-1]{6.5}{11.5}[][black][$r_2$][$d_2$]
    % \interval[-1]{10}{11.5}[][black][$r_2+s_2$]
    % \intervalthin[-2]{10}{10.5}[$\varepsilon$][black]
    % \interval[-2]{10.5}{11.5}[][black][$r_4$][$d_4$]
    \end{tikzpicture}
    \caption{$ALG$ vs $OPT$ schedule}
    \label{fig:advvsopt}
\end{figure}
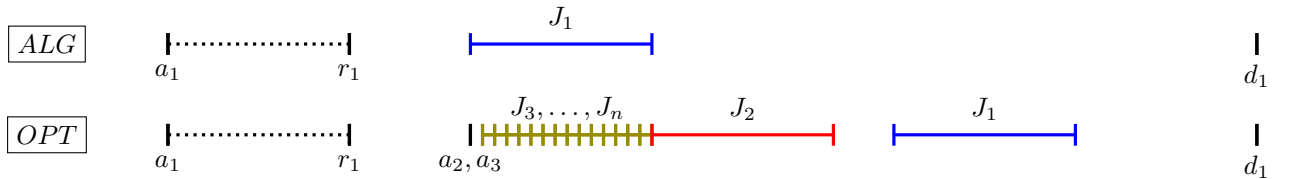

\medskip
\begin{corollary}
    Regardless of how much advance-notice a sequence of jobs is guaranteed, no algorithm can ever have a competitive ratio of $\frac{1}{2}$ or better. 
\end{corollary}

\subsection{Advance Notice without preemption Does Not help Arbitrary C and D Benevolent instances}
    \label{sec:advance-notice-c-and-d}
    Restricting ourselves to proportional weights might in some situations seem too strict. We will now look at the performance of algorithms in cases where we are not dealing with proportional weights and show that against certain \(C\)-benevolent and \(D\)-benevolent input sequences, regardless of $t$, an algorithm running on a set of jobs with $t$-advance-notice can be made to perform arbitrarily poorly and cannot achieve a constant competitive ratio.

The algorithm \(\tau\)-\(\textsc{Persist}^*\) described in \cref{subsec:advance-notice-alg} simulates \(\tau\)-\textsc{Persist}, which performs well against arbitrary C-benevolent and D-benevolent input sequences as well as proportional weights. It is why it may be surprising that advance-notice does not seem to help us remove preemption against arbitrary C-benevolent function.

\medskip
\begin{theorem}
    For any throughput maximization algorithm \(ALG\) that does not use preemption and any $t \in (0,1)$, there exists a C-benevolent function $f$ and a sequence of $f$-related jobs such that any algorithm \(ALG\) can have arbitrarily bad competitive ratio against this sequence of jobs.
\end{theorem}

\begin{proof}
        
    We will show that we can generate a sequence of jobs where \(ALG\)=1 and \(OPT\)=$N$ for any $N$.
    
    We once again use an adversarial argument. Our adversary chooses $f(x) = x^{\log_\frac{1-\varepsilon}{t}(N)}$ to be our function $f$ for some very small epsilon. We note that this function is $C$-benevolent for $N$ sufficiently larger than $t$.
    
    Our adversary first announces $J_1$ with $a_1 = 0$, $r_1 = t$, $p_1 = 1$, $d_1 = t+1$ and $w_1 = f(1)$.
    
    If our algorithm does not take the job the adversary declares it to  be the last, we can assume that the algorithm takes the job.
    
    Our adversary next announces $J_2$ with $a_2 = t+\epsilon$, $r_2 = t+1-2\varepsilon$, $p_2 = \frac{1 - \varepsilon}{t}$ $d_2 = r_2 + p_2$ which the algorithm cannot schedule.
    
    The optimal algorithm can take $J_2$ which has value $f(\frac{1 - \varepsilon}{t}) = (\frac{1 - \varepsilon}{t})^{\log_\frac{1-\varepsilon}{t}(N)} = N$ and the algorithm can only take $J_1$ which has value $f(p_1) = f(1) = 1$. We have an arbitrarily large competitive ratio as desired.
\end{proof}

\begin{theorem}
    For any $t$ and any throughput maximization algorithm \(ALG\) that does not use preemption, there is a sequence of unweighted jobs that has $t$-advanced notice such that \(ALG\) does not achieve a finite competitive ratio against this input sequence. In particular, since \(f(p) = 1\) is a D-benevolent function, there exists a D-benevolent function $f$ and a sequence of $f$-related jobs such that any algorithm \(ALG\) can have arbitrarily bad competitive ratio against this sequence of jobs..
\end{theorem}
\begin{proof}
    Once again we show that we can generate a sequence of jobs where \(ALG=1\) and $OPT=N$ for any $N$.

    Our adversary first announces $J_1$ with $a_1 = 0$, $r_1 = t$, $p_1 = 1$, $d_1 = t+1$ and $w_1 = f(1)$.

    We then announce jobs $J_2,\ldots, J_{N+1}$ such that $\forall i,j \in \{2,\ldots,N+1\}$: $[a_i,d_i] \subset (r_1, d_1)$ and $[a_i,d_i] \cap [a_j,d_j] = \emptyset$. Since all jobs have the same weight, the value of a algorithm that schedules $J_2,\ldots, J_{N+1}$ is $N$, as such \(OPT\) has value at least $N$ but $ALG = 1$ as desired.
\end{proof}

\section{Unweighted Throughput with Revoking}
    \label{sec:unweighted}
    
\subsection{A Positive Result for a Limited Number of  Processing Times}

\medskip
\begin{restatable}{theorem}{klengthedf}\label{thm:k-length-edf}
For the unweighted throughput problem on one machine in the preemption-revoke model with at most $k$ distinct processing times, Algorithm~\ref{alg:k-length-edf} is $1/(2k)$-competitive.
\end{restatable}

Let the distinct processing times be
\(
p_1<p_2<\cdots<p_k.
\)
Algorithm \ref{alg:k-length-edf} always prefers shorter processing times, and uses earliest deadling first (EDF) within each processing time class. The analysis can be seen as a substantial extension of indirect charging as used in \cite{BorodinK23}.

\begin{proof}
Fix an optimal feasible schedule $\textit{OPT}$, and let $\textit{ALG}$ be the set of jobs completed by the algorithm. For each job $j\in \textit{OPT}$, let $\sigma_j$ denote the start time of $j$ in $\textit{OPT}$.

We first define the completed descendant of a job started by the algorithm. If a started job $x$ completes, set $R(x)=x$. If $x$ is revoked and the algorithm next starts a job $y$, set $R(x)=R(y)$. This is well-defined because a job is revoked only when a strictly shorter pending job is available, and hence the processing time strictly decreases along the chain. Therefore every revocation chain eventually ends at a completed job.

We define a charging map $\phi:\textit{OPT}\to \textit{ALG}$ as follows.
\begin{enumerate}
    \item If $j\in \textit{ALG}$, charge $j$ to itself.
    \item If $j\notin \textit{ALG}$, but the algorithm starts $j$ and later revokes it, charge $j$ to $R(j)$.
    \item If the algorithm never starts $j$, let $x$ be the job processed by the algorithm at time $\sigma_j$, and charge $j$ to $R(x)$.
\end{enumerate}
The last rule is well-defined. Since $\textit{OPT}$ starts $j$ at time $\sigma_j$, we have $\sigma_j\le r_j+s_j$, so $j$ is feasible and pending for the algorithm at time $\sigma_j$. Thus the algorithm cannot be idle at time $\sigma_j$.

Now fix a completed job $a\in \textit{ALG}$, and suppose $a$ has processing time $p_i$. Consider the revocation chain ending at $a$:
\(
x_0=a,\;x_1,\ldots,x_q,
\)
where $x_{t+1}$ is the job revoked immediately before the algorithm starts $x_t$, if such a job exists. The processing times in this chain strictly increase:
\(
p(x_0)<p(x_1)<\cdots<p(x_q).
\)
Hence the chain contains at most one job from each class $p_i,p_{i+1},\ldots,p_k$, and therefore
\(
q+1\le k-i+1.
\)

We bound the number of jobs charged to $a$ by considering each chain job $x_t$. There are two possible sources of charges associated with $x_t$.

\smallskip
\noindent
\emph{Type A: the chain job itself.}
If $x_t\in \textit{OPT}$, then $x_t$ contributes at most one charge to $a$: either by self-charge if $x_t=a$, or by the revoked-job rule if $x_t$ is later revoked. Thus the chain contributes at most $q+1$ Type A charges.

\smallskip
\noindent
\emph{Type B: never-started jobs charged through $x_t$.}
Let $I_t$ be the interval during which the algorithm processes $x_t$ before $x_t$ either completes or is revoked. Suppose $x_t$ has processing time $p_h$. Then $|I_t|\le p_h$, with strict inequality if $x_t$ is revoked.

While the algorithm is processing $x_t$, no pending job of processing time smaller than $p_h$ exists; otherwise the shorter-length priority rule would revoke $x_t$. Therefore any never-started job $j\in \textit{OPT}$ charged through $x_t$ must have processing time at least $p_h$. Since $\textit{OPT}$ is feasible, at most one job of processing time at least $p_h$ can have its $\textit{OPT}$-start time inside the interval $I_t$, whose length is at most $p_h$. Thus each chain job $x_t$ contributes at most one Type B charge, and the chain contributes at most $q+1$ Type B charges in total.

Combining the two types, the number of jobs charged to $a$ is at most
\[
|\phi^{-1}(a)|\le 2(q+1)\le 2(k-i+1)\le 2k.
\]
Therefore every completed job of the algorithm receives at most $2k$ charges, and so
\[
|\textit{OPT}|
=\sum_{a\in \textit{ALG}}|\phi^{-1}(a)|
\le \sum_{a\in \textit{ALG}} 2k
=2k|\textit{ALG}|.
\]
Equivalently,
\[
|\textit{ALG}|\ge \frac{1}{2k}|\textit{OPT}|.
\]
Thus Algorithm~\ref{alg:k-length-edf} is $1/(2k)$-competitive.
\end{proof}

\begin{algorithm}[h]
\caption{$k$-Length EDF with Shorter-Length Priority}
\label{alg:k-length-edf}
\begin{algorithmic}[1]
    \State $J_{curr}\gets$ None
    \Comment {Job currently being processed}
    \For {$i=1,\ldots,k$}
        \State $Q_i\gets \emptyset$
        \Comment {Pending jobs with processing time $p_i$, ordered by EDF}
    \EndFor
    \State $ALG\gets \emptyset$
    \Comment {Set of jobs completed by the algorithm}
    \State $T\gets 0$
    \Comment {Current time}
    \\
    \While {there are more jobs to arrive, or $J_{curr}$ is not None, or some $Q_i\neq \emptyset$}
        \State Advance to the next release time or the completion time of $J_{curr}$
        \If {$J_{curr}$ completes at the current time $T$}
            \State $ALG\gets ALG\cup\{J_{curr}\}$
            \State $J_{curr}\gets$ None
        \EndIf
        \\
        \For {$i=1,\ldots,k$}
            \State Remove from $Q_i$ all jobs $J_j$ with $T>r_j+s_j$
            \Comment {Expired}
        \EndFor
        \For {$i=1,\ldots,k$}
            \State $A_i\gets$ set of jobs with processing time $p_i$ arriving at time $T$
            \State Insert all jobs in $A_i$ into $Q_i$ in EDF order
        \EndFor
        \\
        \If {$J_{curr}$ has processing time $p_h$ and $Q_i\neq \emptyset$ for some $i<h$}
            \State $J_{curr}\gets$ None
            \Comment {revoke the current job for a shorter pending job}
        \EndIf
        \\
        \If {$J_{curr}$ is None}
            \If {some $Q_i\neq \emptyset$}
                \State Let $i$ be the smallest index such that $Q_i\neq \emptyset$
                \State $J_j\gets$ extract\_min($Q_i$)
                \Comment {EDF within the shortest nonempty class}
                \State $J_{curr}\gets J_j$
            \EndIf
        \EndIf
    \EndWhile
\end{algorithmic}
\end{algorithm}

\subsection{Impossibilty Result for Bounded Processing Times}

For unweighted throughput on one machine in the preemption-revoke model, we prove that for every $k\ge 1$ there is an instance with at most $k$ distinct processing times on which any deterministic online algorithm completes at most one job while the offline optimum completes $k+1$ jobs. Hence no deterministic algorithm can guarantee a competitive ratio greater than $1/(k+1)$ on instances with at most $k$ processing times, and no deterministic algorithm can achieve a constant competitive ratio when the number of processing times is unrestricted.

For simplicity, we will represent an unweighted job \(J_i = (r_i,p_i,d_i,w_i)\) as \(J_i = (r_i, p_i,s_i)\), where \(s_i\) is the slack of the job \(J_i\).

The lower-bound construction is based on a recursive construction of adversarial gadgets. A gadget starts at a time chosen by the adversary, releases jobs adaptively according to the behavior of $\textit{ALG}$, and has a bounded time span. The key point is that a larger processing time can be used as a ``bait job'' that hides a complete smaller gadget inside its execution interval. We begin with the equal-length terminal gadget and then give the general nested construction.

\textbf{One Distinct Processing Time}

The base case ($k = 1$). 
%We first record the equal-length case, which is the terminal gadget for the general construction. 
We present the construction with a arbitrary processing time $q>0$.

\medskip
\begin{restatable}{theorem}{equallengthrevoke}\label{thm:equal-length-revoke}
For the unweighted throughput problem on one machine in the preemption-revoke model for equal-length jobs, no deterministic online algorithm can guarantee a competitive ratio greater than $1/2$.
\end{restatable}
\begin{proof}
Fix any deterministic online algorithm $\textit{ALG}$, and release a job
\[
A=(r_A,\;q,\;2q).
\]
If $\textit{ALG}$ starts $A$ at some time $\alpha\in [r_A,r_A+q)$, then release an urgent job
\[
B=(\alpha+\varepsilon,\;q,\;0),
\]
where $\varepsilon>0$ is sufficiently small. If $\textit{ALG}$ continues processing $A$, then it misses $B$. If it revokes $A$ and processes $B$, then $A$ is lost forever. Hence $\textit{ALG}$ completes at most one job. In contrast, $\textit{OPT}$ completes both jobs by scheduling $B$ in $[\alpha+\varepsilon,\alpha+\varepsilon+q)$ and $A$ in $[r_A+2q,r_A+3q)$.

If $\textit{ALG}$ starts $A$ at some time $\alpha\in [r_A+q,r_A+2q)$, then release
\[
B=(\alpha+\varepsilon,\;q,\;0),
\]
again with $\varepsilon>0$ sufficiently small. The algorithm completes at most one of $A$ and $B$, while $\textit{OPT}$ completes both jobs by scheduling $A$ in $[r_A,r_A+q)$ and $B$ in $[\alpha+\varepsilon,\alpha+\varepsilon+q)$.

Finally, if $\textit{ALG}$ has not started $A$ before time $r_A+2q$, then release
\[
B=(r_A+2q,\;q,\;0).
\]
At time $r_A+2q$, the algorithm can complete at most one of $A$ and $B$, while $\textit{OPT}$ completes $A$ in $[r_A,r_A+q)$ and $B$ in $[r_A+2q,r_A+3q)$.

Thus in every case $\textit{ALG}\le 1$ and $\textit{OPT}=2$, proving the theorem.
\end{proof}

The same terminal gadget can be nested inside successively larger bait jobs. We choose processing times
\[
\ell_1<\ell_2<\cdots<\ell_k
\]
recursively. Let $T_1=3\ell_1$. For each $i\ge 2$, choose $\ell_i>T_{i-1}$ and define
\[
T_i=2\ell_i+T_{i-1}.
\]
The number $T_i$ is the span reserved for the $i$-level gadget.

\medskip
\begin{theorem}\label{thm:k-length-revoke-lower}\label{thm:negative-unweighted}
For every $k\ge 1$, there exists an instance with at most $k$ distinct processing times such that $\textit{ALG}$ completes at most one job while $\textit{OPT}$ completes at least $k+1$ jobs.
Therefore, no deterministic online algorithm can guarantee a competitive ratio greater than $1/(k+1)$ on instances with at most $k$ distinct processing times.
\end{theorem}

\begin{proof}
We prove the following stronger inductive statement. For each $i=1,\ldots,k$ and each start time $a$, there is an adversarial gadget $G_i(a)$ contained in the interval $[a,a+T_i)$, using only processing times from $\{\ell_1,\ldots,\ell_i\}$, such that $\textit{ALG}\le 1$ while $\textit{OPT}=i+1$ on the jobs released by the gadget.

For $i=1$, the gadget $G_1(a)$ is exactly the equal-length construction from Theorem~\ref{thm:equal-length-revoke}, with $r_A=a$ and $q=\ell_1$. It is contained in $[a,a+3\ell_1)=[a,a+T_1)$, uses one processing time, and forces $\textit{ALG}\le 1$ while $\textit{OPT}=2$.

Now assume the claim holds for $i-1$, where $i\ge 2$. We construct $G_i(a)$. At time $a$, release a bait job
\[
L_i=(a,\;\ell_i,\;\ell_i+T_{i-1}).
\]
The latest feasible start time of $L_i$ is $a+\ell_i+T_{i-1}$, and the end of its window is
\[
a+2\ell_i+T_{i-1}=a+T_i.
\]

\medskip
\noindent
\textbf{Case 1.}
Suppose $\textit{ALG}$ starts $L_i$ at some time $t\in [a,a+\ell_i)$. Immediately after this, release a shifted copy $G_{i-1}(t+\varepsilon)$, where $\varepsilon>0$ is chosen sufficiently small so that
\[
t+\varepsilon+T_{i-1}<a+\ell_i+T_{i-1}
\qquad\text{and}\qquad
\varepsilon+T_{i-1}<\ell_i.
\]
If $\textit{ALG}$ keeps processing $L_i$, then it is busy throughout the whole lower-level gadget, so it completes at most $L_i$. If $\textit{ALG}$ revokes $L_i$, then $L_i$ is lost forever, and by the inductive hypothesis the lower-level gadget gives $\textit{ALG}$ at most one completed job. Thus $\textit{ALG}\le 1$.

The optimal schedule completes the $i$ jobs of $G_{i-1}(t+\varepsilon)$ first and then schedules $L_i$ late, starting at $a+\ell_i+T_{i-1}$. This is feasible because the first displayed inequality ensures that the lower-level gadget finishes before the latest feasible start time of $L_i$. Hence $\textit{OPT}=i+1$.

\medskip
\noindent
\textbf{Case 2.}
Suppose $\textit{ALG}$ starts $L_i$ at some time
\[
t\in [a+\ell_i,\;a+\ell_i+T_{i-1}].
\]
Immediately after this, release $G_{i-1}(t+\varepsilon)$, where $\varepsilon>0$ is chosen sufficiently small so that
\(
\varepsilon+T_{i-1}<\ell_i.
\)
If $\textit{ALG}$ keeps processing $L_i$, then it is busy throughout the whole lower-level gadget. If it revokes $L_i$, then $L_i$ is lost forever and the lower-level gadget gives it at most one completed job. Hence $\textit{ALG}\le 1$.

The optimal schedule runs $L_i$ early in $[a,a+\ell_i)$ and then completes the $i$ jobs of $G_{i-1}(t+\varepsilon)$. Since $t+\varepsilon>a+\ell_i$, these schedules do not overlap. Therefore $\textit{OPT}=i+1$.

\medskip
\noindent
\textbf{Case 3.}
Suppose $\textit{ALG}$ has not started $L_i$ before time $a+\ell_i+T_{i-1}$. At that time, release $G_{i-1}(a+\ell_i+T_{i-1})$. If $\textit{ALG}$ starts $L_i$ and keeps it, then it is busy throughout the lower-level gadget because $\ell_i>T_{i-1}$. If $\textit{ALG}$ revokes $L_i$, or does not start it, then $L_i$ is lost or infeasible, and the lower-level gadget gives $\textit{ALG}$ at most one completed job. Thus $\textit{ALG}\le 1$.

The optimal schedule runs $L_i$ early in $[a,a+\ell_i)$ and then completes the $i$ jobs of $G_{i-1}(a+\ell_i+T_{i-1})$. Hence $\textit{OPT}=i+1$.

In all three cases the lower-level gadget is contained in the reserved interval $[a,a+T_i)$: the only nontrivial condition is that its span $T_{i-1}$ must fit inside the execution interval of $L_i$ or inside the remaining part of $[a,a+T_i)$, which follows from $\ell_i>T_{i-1}$ and the choice of $\varepsilon$.

This completes the induction. Taking $i=k$ gives an instance with at most $k$ distinct processing times, $\textit{ALG}\le 1$, and $\textit{OPT}=k+1$. Therefore
\(
\frac{\textit{ALG}}{\textit{OPT}}\le \frac{1}{k+1}.
\)
\end{proof}

\medskip
% \begin{corollary}\label{cor:no-constant-revoke}
% No deterministic online algorithm can achieve a constant competitive ratio for unweighted throughput on one machine in the preemption-revoke model when the number of distinct processing times is unrestricted.
% \end{corollary}

% \begin{proof}
% For every $k\ge 1$, Theorem~\ref{thm:k-length-revoke-lower} gives an instance on which $\textit{ALG}/\textit{OPT}\le 1/(k+1)$. This ratio tends to zero as $k$ grows.
% \end{proof}

\section{Conclusion}
    \label{sec:conclusion}
    
We have presented a number of results for the throughput maximization  problem for unweighted and weighted instances. These results are compared with results for the interval selection problem. Our results also raise some interesting open questions. 

The Woeginger \cite{Woeginger94} deterministic competitive ratio for interval selection for C-Benevolent and D-Benevolent  instances was adapted to the throughput problem resulting in a $\frac{1}{5}$ deterministic competitive ratio. The optimality of the $\frac{1}{4}$ ratio for interval selection implies the same $\frac{1}{4}$ limitation for the 
throughput problem. What is the optimal competitive ratio for C-Benevolent weight functions? The extension from interval selection to throughput is essentially achieved by replacing revoking by restarting and then maintaining a queue  of jobs that can still be executed by their deadlines. This same idea and charging analysis can be seen as a way to derive the Hoogeveen et al unweighted result. In terms of ratios $\rho \geq 1$, the ratio $\rho$ for interval selection becomes $\rho' = \rho + 1$ for the analogous throughput problem. Can this idea be applied to other real-time problems?  
%Can  say that a similar  idea is what is giving the 1/2k posiitve results for unweighted throughpiutb genertaliing unweighted interval selection? For the unweuighted case, 

We have assumed the clairvoyant hard deadline model for the throughput problem. It would be interesting to study alternative assumptions for the throughput problem with restarting such as partial profit for incomplete jobs or costs for exceeding deadlines. 

We have introduced a new real-time model in which an algorithm is guaranteed some amount of ``advance notice''. In this model, with at least $t_i = p_i/2$ advance notice, we can realize the same $\frac{1}{5}$ competitive ratio for the throughput algorithm for proportional instances without using preemption. Can we utilize advance notice for other real-time problems so as to eliminate the need for preemption? What if there is a cost for insufficient advance notice? 

In contrast to interval selection, it is not possible to have a constant competitive algorithm {\it with revoking} for the unweighted throughput problem when the number of distinct processing times is unrestricted. For instances with at most $k$ distinct processing times, our lower bound shows that no deterministic real-time algorithm with revoking can obtain a ratio better than $\frac{1}{k+1}$. On the positive side, our $k$-length EDF algorithm is $\frac{1}{2k}$-competitive. Thus the bounded-processing-time case admits a ratio depending only on $k$, while the unrestricted case does not admit any constant competitive ratio. Closing the gap between $\frac{1}{k+1}$ and $\frac{1}{2k}$ remains open.

While our focus has been on deterministic algorithms, we ask which of  our negative and positive results can be strengthened by allowing  randomization. In particular, to what extent can even 1-bit of randomness lead to improved competitive ratios. The Fung et al \cite{fung_improved_2014} and Chrobak et al \cite{chrobak_online_2007} barely random algorithms for the single processing time weighted and unweighted throughout problems suggest that randomization will lead to improved competitive ratios for other variants of the throughput problem. 
%With regard to the unweighted case, Chrobak et al \cite{ChrobakJST07} provide an optimmlly competitive algorithm in the real-time (with restarting) model. Is there a  unweighted single length instances. 

\section*{Acknowledgements}

We thank Andrew Tam for his insightful comments which led to a clearer presentation of our analysis of the \(\tau\)-\textsc{Persist} algorithm for C-benevolent instances.
    \newpage
    \bibliographystyle{plain}
    \bibliography{ref} 

@inproceedings{BorodinK23,
  author       = {Allan Borodin and
                  Christodoulos Karavasilis},
  editor       = {Jaroslaw Byrka and
                  Andreas Wiese},
  title        = {Any-Order Online Interval Selection},
  booktitle    = {Approximation and Online Algorithms - 21st International Workshop,
                  {WAOA} 2023, Amsterdam, The Netherlands, September 7-8, 2023, Proceedings},
  series       = {Lecture Notes in Computer Science},
  volume       = {14297},
  pages        = {175--189},
  publisher    = {Springer},
  year         = {2023},
  }

@article{VeselyCJS22,
  author       = {Pavel Vesel{\'{y}} and
                  Marek Chrobak and
                  Lukasz Jez and
                  Jir{\'{\i}} Sgall},
  title        = {A {\textbackslash}({\textbackslash}boldsymbol\{{\textbackslash}phi
                  \}{\textbackslash}) -Competitive Algorithm for Scheduling Packets
                  with Deadlines},
  journal      = {{SIAM} J. Comput.},
  volume       = {51},
  number       = {5},
  pages        = {1626--1691},
  year         = {2022}
}

@inproceedings{FaigleN91,
  author        = {Ulrich Faigle and
               Willem M. Nawijn}, 
  title         = {Greedy k-decomposition of interval orders},
  booktitle     = {Second Twente Workshop on Graphs and Combinatorial Optimization},
  pages         = {53-56},
  year          = {1991},
}

@article{karamata,
 author = {Karamata, J.},
 title = {Sur une in{\'e}galit{\'e} relative aux fonctions convexes.},
 fjournal = {Publications Math{\'e}matiques de l'Universit{\'e} de Belgrade},
 journal = {Publ. Math. Univ. Belgrade},
 volume = {1},
 pages = {145--148},
 year = {1932},
 language = {French},
 zbMATH = {2549602},
 JFM = {58.0211.01}
}

@article{BaruahKMMRRSW92,
  author        = {Sanjoy K. Baruah and
               Gilad Koren and
               Decao Mao and
               Bhubaneswar Mishra and
               Arvind Raghunathan and
               Louis E. Rosier and
               Dennis E. Shasha and
               Fuxing Wang},
  title         = {On the Competitiveness of On-Line Real-Time Task Scheduling},
  journal       = {RealTimeSyst},
  volume        = {4},
  number        = {2},
  pages         = {125--144},
  year          = {1992},
}

@article{ChinFGHHJLTZZZ15,                                                                                                                                                                              
        title = {Online competitive algorithms for maximizing weighted throughput of unit jobs},                                                                                                             
        volume = {4},                                                                                                                                                                                        
        issn = {1570-8667},                                                                                                                                                                                  
        url = {https://www.sciencedirect.com/science/article/pii/S1570866705000250},                                                                                                                         
        doi = {10.1016/j.jda.2005.03.005},                                                                   
        number = {2},                                                                                                                                                                                        
        urldate = {2026-06-25},                                                                                                                                                                              
        journal = {Journal of Discrete Algorithms},                                                                                                                                                          
        author = {Chin, Francis Y. L. and Chrobak, Marek and Fung, Stanley P. Y. and Jawor, Wojciech and Sgall, Jiří and Tichý, Tomáš},                                                                      
        month = jun,                                                                                                                                                                                         
        year = {2006},                                                                                                                                       
        pages = {255--276},                                                                                
}

@inproceedings{KorenS92,
  author        = {Gilad Koren and
               Dennis E. Shasha},
  title         = {D\({}^{\mbox{over}}\); an optimal on-line scheduling algorithm for
               overloaded real-time systems},
  booktitle     = {Proceedings of the Real-Time Systems Symposium - 1992, Phoenix, Arizona,     
               USA, December 1992},
  pages         = {290--299},
  year          = {1992}, 
  publisher     = {{IEEE} Computer Society},
}

@article{KalyanasundaramP03,
  author        = {Bala Kalyanasundaram and
                  Kirk Pruhs},
  title         = {Maximizing job completions online},
  journal       = {JAlgorithms},
  volume        = {49},
  number        = {1},
  pages         = {63--85},
  year          = {2003},
}

@inproceedings{BaruhaHS94,
  author        = {Sanjoy K. Baruah and
                  Jayant R. Haritsa and
                  Nitin Sharma},
  title         = {On-Line Scheduling to Maximize Task Completions},
  booktitle     = {Proceedings of the 15th {IEEE} Real-Time Systems Symposium {(RTSS}              
                  '94), San Juan, Puerto Rico, December 7-9, 1994},
  pages         = {228--236},
  year          = {1994},
  publisher     = {{IEEE} Computer Society},
}

@inproceedings{LucierMNY13,
  author       = {Brendan Lucier and
                  Ishai Menache and
                  Joseph Naor and
                  Jonathan Yaniv},
  editor       = {Guy E. Blelloch and
                  Berthold V{\"{o}}cking},
  title        = {Efficient online scheduling for deadline-sensitive jobs: extended
                  abstract},
  booktitle    = {25th {ACM} Symposium on Parallelism in Algorithms and Architectures,
                  {SPAA} '13, Montreal, QC, Canada - July 23 - 25, 2013},
  pages        = {305--314},
  publisher    = {{ACM}},
  year         = {2013},
}

@article{chrobak_online_2007,
	title = {Online {Scheduling} of {Equal}‐{Length} {Jobs}: {Randomization} and {Restarts} {Help}},
	volume = {36},
	issn = {0097-5397, 1095-7111},
	shorttitle = {Online {Scheduling} of {Equal}‐{Length} {Jobs}},
	url = {http://epubs.siam.org/doi/10.1137/S0097539704446608},
	doi = {10.1137/S0097539704446608},
	language = {en},
	number = {6},
	urldate = {2026-03-27},
	journal = {SIAM Journal on Computing},
	author = {Chrobak, Marek and Jawor, Wojciech and Sgall, Jiří and Tichý, Tomáš},
	month = jan,
	year = {2007},
	pages = {1709--1728},
}

@article{GarayGKMY97,
  author        = {Juan A. Garay and
                  Inder S. Gopal and
                  Shay Kutten and
                  Yishay Mansour and
                  Moti Yung},
  title         = {Efficient On-Line Call Control Algorithms},
  journal       = {Journal of Algorithms},
  volume        = {23},
  number        = {1},
  pages         = {180--194},
  year          = {1997},
}

@article{HoogeveenPW00,
	title = {On-line scheduling on a single machine: maximizing the number of early jobs},
	volume = {27},
	copyright = {https://www.elsevier.com/tdm/userlicense/1.0/},
	issn = {01676377},
	shorttitle = {On-line scheduling on a single machine},
	url = {https://linkinghub.elsevier.com/retrieve/pii/S0167637700000614},
	doi = {10.1016/S0167-6377(00)00061-4},
	language = {en},
	number = {5},
	urldate = {2026-03-27},
	journal = {Operations Research Letters},
	author = {Hoogeveen, Han and Potts, Chris N. and Woeginger, Gerhard J.},
	month = dec,
	year = {2000},
	pages = {193--197},
}

@article{EmekHR2016,
	title = {Space-{Constrained} {Interval} {Selection}},
	volume = {12},
	issn = {1549-6325, 1549-6333},
	url = {https://dl.acm.org/doi/10.1145/2886102},
	doi = {10.1145/2886102},
	language = {en},
	number = {4},
	urldate = {2026-03-27},
	journal = {ACM Transactions on Algorithms},
	author = {Emek, Yuval and Halldórsson, Magnús M. and Rosén, Adi},
	month = sep,
	year = {2016},
	pages = {1--32},
}

@inproceedings{CanettiI98,
  title={Bounding the power of preemption in randomized scheduling},
  author={Canetti, Ran and Irani, Sandy},
  booktitle={Proceedings of the twenty-seventh annual ACM symposium on Theory of computing},
  pages={606--615},
  year={1995}
}

@article{fung_improved_2014,
	title = {Improved {Randomized} {Online} {Scheduling} of {Intervals} and {Jobs}},
	volume = {55},
	issn = {1433-0490},
	url = {https://doi.org/10.1007/s00224-013-9528-2},
	doi = {10.1007/s00224-013-9528-2},
	language = {en},
	number = {1},
	urldate = {2026-04-02},
	journal = {Theory of Computing Systems},
	author = {Fung, Stanley P. Y. and Poon, Chung Keung and Zheng, Feifeng},
	month = jul,
	year = {2014},
	pages = {202--228},
}

@article{EpsteinL10,
  author        = {Leah Epstein and 
               Asaf Levin},
  title         = {Improved randomized results for the interval selection problem},
  journal       = {Theoretical Computer Science},
  volume        = {411},
  number        = {34-36},
  pages         = {3129--3135},
  year          = {2010},
}

@article{goldman_online_2000,
	title = {Online {Scheduling} with {Hard} {Deadlines}},
	volume = {34},
	issn = {0196-6774},
	url = {https://www.sciencedirect.com/science/article/pii/S019667749991060X},
	doi = {10.1006/jagm.1999.1060},
	number = {2},
	urldate = {2026-04-23},
	journal = {Journal of Algorithms},
	author = {Goldman, Sally A and Parwatikar, Jyoti and Suri, Subhash},
	month = feb,
	year = {2000},
	pages = {370--389},
}

@article{Woeginger94,
	title = {On-line scheduling of jobs with fixed start and end times},
	volume = {130},
	issn = {0304-3975},
	url = {https://www.sciencedirect.com/science/article/pii/0304397594901503},
	doi = {10.1016/0304-3975(94)90150-3},
	number = {1},
	urldate = {2026-04-28},
	journal = {Theoretical Computer Science},
	author = {Woeginger, Gerhard J.},
	month = aug,
	year = {1994},
	pages = {5--16},
}

@incollection{Lenstra77,
  author        = {J. K. Lenstra and A. H. G. Rinnooy Kan and P. Brucker},
  title         = {Complexity of Machine Scheduling Problems},
  booktitle     = {Annals of Discrete Mathematics},
  volume        = {1},
  pages         = {343--362},
  year          = {1977},
}

@article{Lawler90,
  author        = {E. L. Lawler},
  title         = {A Dynamic Programming Algorithm for Preemptive Scheduling of a Single Machine to Minimize the Number of Late Jobs},
  journal       = {Annals of Operations Research},
  volume        = {26},
  number        = {1},
  pages         = {125--133},
  year          = {1990},
}

@article{Baptiste99,
  author        = {Philippe Baptiste},
  title         = {An {$O(n^4)$} Algorithm for Preemptive Scheduling of a Single Machine to Minimize the Number of Late Jobs},
  journal       = {Operations Research Letters},
  volume        = {24},
  number        = {3},
  pages         = {175--180},
  year          = {1999},
}

@inproceedings{Lipton,
  author        = {Richard J. Lipton and Andrew Tomkins},
  title         = {Online Interval Scheduling},
  booktitle     = {Proceedings of the Fifth Annual {ACM-SIAM} Symposium on Discrete Algorithms},
  pages         = {302--311},
  year          = {1994},
}
\end{document}